\documentclass[a4paper,american]{scrartcl}

\usepackage{tikz-cd}
\usepackage{verbatim}
\usepackage{enumerate}
\usepackage{lmodern}
\usepackage[utf8]{inputenc}
\usepackage[T1]{fontenc}
\usepackage{color}
\usepackage{todonotes}
\usepackage{babel}
\usepackage{amsfonts, amsmath, amsthm, amssymb}
\usepackage{setspace}
\usepackage[font=small]{quoting}
\usepackage{latexsym}
\usepackage[bookmarks=false,pdfborder={0 0 0},colorlinks=false]{hyperref}
\usepackage{cite}
\usepackage{dsfont}
\usepackage{xcolor}
\makeatletter
\renewcommand\@biblabel[1]{}
\makeatother

\renewcommand{\vec}{\boldsymbol}

\usepackage[arrow, matrix, curve]{xy}

\newcommand{\IR}{\mathbb R}

\newcommand{\IP}{\mathbb P}

\newcommand{\dd}{\mathrm{d}}

\theoremstyle{plain}
\newtheorem{Theorem}{Theorem}[section]

\newtheorem{Corollary}[Theorem]{Corollary}
\newtheorem{Proposition}[Theorem]{Proposition}
\newtheorem*{Proposition*}{Proposition}
\theoremstyle{definition}

\title{\LARGE
Arrow(s) of Time without a Past Hypothesis
}
  \author{Dustin Lazarovici\thanks{Universit\'e de Lausanne, Section de Philosophie. Email: dustin.lazarovici@live.com}\;  and Paula Reichert\thanks{LMU München, Mathematisches Institut. Email: reichert@math.lmu.de}}

\begin{document}
	\maketitle
		\begin{abstract}\noindent The paper discusses recent proposals by Carroll and Chen, as well as Barbour, Koslowski, and Mercati to explain the (thermodynamic) arrow of time without a Past Hypothesis, i.e., the assumption of a special (low-entropy) initial state of the universe. After discussing the role of the Past Hypothesis and the controversy about its status, we explain why Carroll's model -- which establishes an arrow of time as typical -- can ground sensible predictions and retrodictions without assuming something akin to a Past Hypothesis. We then propose a definition of a Boltzmann entropy for a classical $N$-particle system with gravity, suggesting that a Newtonian gravitating universe might provide a relevant example of Carroll's entropy model. This invites comparison with the work of Barbour, Koslowski, and Mercati that identifies typical arrows of time in a relational formulation of classical gravity on shape space. We clarify the difference between this gravitational arrow in terms of shape complexity and the entropic arrow in absolute spacetime, and work out the key advantages of the relationalist theory. We end by pointing out why the entropy concept relies on absolute scales and is thus not relational. 
			
			 \end{abstract}
	\section{The easy and the hard problem of irreversibility}

What is the difference between past and future? Why do so many physical processes occur in only one time direction, despite the fact that they are governed or described, on the fundamental level, by time-symmetric microscopic laws? These questions are intimately linked to the notion of \emph{entropy} and the second law of thermodynamics. From the point of view of fundamental physics, it is the second law of thermodynamics that accounts for such phenomena as that gases expand rather than contract, that glasses break but don't spontaneously reassemble, that heat flows from hotter to colder bodies, that a car slows down and doesn't accelerate once you stop hitting the gas. All these are examples of irreversible processes, associated with an increase of entropy in the relevant physical systems. 
	
Goldstein (2001) -- possibly inspired by Chalmers' discussion of the mind-body problem (Chalmers, 1995) -- distinguishes between the \emph{easy part} and the \emph{hard part} of the problem of irreversibility. The easy part of the problem is: \emph{Why do isolated systems in a state of low entropy typically evolve into states of higher entropy (but not the other way round)?} The answer to this question was provided by Ludwig Boltzmann, who reduced the second law of thermodynamics to the statistical mechanics of point particles. He thereby developed insights and concepts whose relevance goes beyond the confines of any particular microscopic theory. 

The first crucial concept of Boltzmann's statistical mechanics is the distinction between micro- and macrostates. Whereas the microstate $X(t)$ of a system is given by the complete specification of its microscopic degrees of freedom, its macrostate $M(t)$ is specified by (approximate) values of ``observables'' that characterize the system on macroscopic scales (typical examples are volume, pressure, temperature, magnetization, and so on). The macroscopic state of a system is completely determined by its microscopic configuration, that is $M(t) = M(X(t))$, but one and the same macrostate can be realized by a large (in general infinite) number of different microstates, all of which ``look macroscopically the same''. Partitioning the microscopic state space into sets corresponding to macroscopically distinct configurations is therefore called \textit{coarse-graining}. Turning to the phase space picture of Hamiltonian mechanics for an $N$-particle system, a microstate corresponds to one point $X=(q,p)$ in phase space $\Omega \cong \IR^{3N}\times\IR^{3N}$, $q=(\mathbf{q}_1, \ldots, \mathbf{q}_N)$ being the position- and $p=(\mathbf{p}_1, \ldots, \mathbf{p}_N)$ the momentum-coordinates of the particles, whereas a macrostate $M$ corresponds to an entire region $\Gamma_M \subseteq \Omega$ of phase space (``macroregion''), namely the set of all microstates coarse-graining to $M$. Boltzmann then realized that for macroscopic systems -- that is, systems with a very large numbers of microscopic degrees of freedom -- different macrostates will in general correspond to macroregions of vastly different size, as measured by the pertinent stationary phase space measure (the Lebesgue- or Liouville-measure in case of a classical Hamiltonian system), with the equilibrium state corresponding to the macroregion of by far largest measure, exhausting almost the entire phase space volume. The Boltzmann entropy of a system is now defined as the logarithm of the phase space volume covered by its current macroregion (times a dimensional constant called Boltzmann constant): 
\begin{equation}
S(X(t)) = k_B\, \log |\Gamma_{M(X(t))} |.
\end{equation}
 Since the entropy is an extensive macrovariable (proportional to $N$, the number of microscopic constituents), we see that the ratio of phase space volume corresponding to macroregions of significantly different entropy values is of order $\exp(N)$, where $N\sim 10^{24}$ even for ``small'' macroscopic systems (the relevant order of magnitude is given by Avogadro's constant).  
 We thus understand why, under the chaotic microdynamics of a many-particle system, a microstate starting in a small (low-entropy) region of phase space will typically evolve into larger and larger macroregions, corresponding to higher and higher entropy, until it reaches the equilibrium state where it will spend almost the entire remainder of its history (that is, apart from rare fluctuations into lower-entropy states). Taken with enough grains of salt, we can summarize this in the slogan that an irreversible (i.e., entropy increasing) process corresponds to an evolution from less likely into more likely macrostates. 
  
  Notably, the time-reversal invariance of the microscopic laws implies that it cannot be true that \emph{all} microstates in a low-entropy macroregion $\Gamma_{M_1}$ evolve into states of higher entropy. But microconditions leading to an entropy-decreasing evolution are \emph{atypical} -- they form a subset of extremely small measure -- while nearly all microstates in $\Gamma_{M_1}$ evolve into states of higher entropy, i.e., entropy increase is \emph{typical}.

The easy problem of irreversibility can be arbitrarily hard from a technical point of view if one seeks to obtain rigorous mathematical results about the convergence to equilibrium in realistic physical models. It is easy in the sense that, conceptually, Boltzmann's account is well understood and successfully applied in physics and mathematics -- despite ongoing (but largely unnecessary) controversies and misconceptions in the philosophical literature (see Bricmont (1999) and Lazarovici and Reichert (2015) for a more detailed discussion and responses to common objections).

The hard problem begins with the question: \emph{Why do we find systems in low-entropy states to begin with if these states are so unlikely?} Often the answer is that \emph{we} prepared them, creating low-entropy subsystems for the price of increasing the entropy in their environment. But why then is the entropy of this environment so low -- most strikingly in the sense that it allows \emph{us} to exist? If one follows this rationale to the end, one comes to the conclusion that the universe as a whole is in a state of low entropy (that is, globally, in a spatial sense; we don't just find ourselves in a low-entropy pocket in an otherwise chaotic universe) and that this state must have evolved from a state of even lower entropy in the distant past. The latter assumption is necessary to avoid the absurd conclusion that our present macrostate -- which includes all our memories and records of the past -- is much more likely the product of a fluctuation out of equilibrium than of the low-entropy past that our memories and records actually record. In other words: only with this assumption does Boltzmann's account ``make it plausible not only that the paper will be yellower and ice cubes more melted and people more aged and smoke more dispersed in the future, but that they were less so (just as our experience tells us) in the past.'' (Albert (2015, p. 5); for a good discussion of this issue see also Feynman (1967, Ch. 5), and Carroll (2010).)

In sum, the hard part of the problem of irreversibility is to explain the existence of a \emph{thermodynamic arrow of time in our universe}, given that the universe is governed, on the fundamental level, by reversible microscopic laws. And the standard account today involves the postulate of a very special (since very low-entropy) initial macrostate of the universe. Albert (2001) coined for this postulate the now famous term \emph{Past Hypothesis} (PH). But the status of the Past Hypothesis is highly controversial. Isn't the very low-entropy beginning of the universe itself a mystery in need of scientific explanation?

In the next section, we will briefly recall this controversy and the various attitudes taken towards the status of the PH. Section \ref{sec:SLwoPH} then introduces recent ideas due to Sean Carroll and Julian Barbour to explain the arrow of time \emph{without} a Past Hypothesis, that is, as a feature of \emph{typical} universes. In Section \ref{sec:rolePH}, we discuss if the Carroll model can indeed ground sensible inferences about the past and future of our universe without assuming something akin to the PH. Section \ref{sec:graventropy} will propose a Boltzmann entropy for Newtonian gravity, suggesting that a universe of $N$ gravitating point particles provides a relevant realization of Carroll's entropy model. Section \ref{sec:barbour}, finally, introduces the results of Barbour, Koslowski, and Mercati that establishes an arrow of time in a relational formulation of Newtonian gravity on ``shape space''. We will work out potential advantages of the relational theory and clarify the differences between the ``gravitional arrow'' of Barbour et al. and the entropic arrow identified in the preceding section.

\section{The controversy over the Past Hypothesis}
So, the standard response to the hard part of the problem of irreversibility involves the postulate of a very low-entropy beginning of our universe. But what is the status of this Past Hypothesis? In the literature, by and large three different views have been taken towards this issue.
\begin{enumerate}
	\item The low-entropy beginning of the universe requires an explanation. 
	\item The low-entropy beginning of the universe does not require, or allow, any further explanation. 
	\item The Past Hypothesis is a law of nature (and therefore does not require or allow any further explanation). 
\end{enumerate}

	The first point of view is largely motivated by the fact that our explanation of the thermodynamic arrow is based on a \emph{typicality reasoning} (e.g. Lebowitz (1993), Goldstein (2001), Lazarovici and Reichert (2015)). Assuming a low-entropy initial macrostate of the universe, Boltzmann's analysis allows us to conclude that \emph{typical} microstates relative to this macrostate will lead to a thermodynamic evolution of increasing entropy. It is then not a good question to ask why the actual initial conditions of the universe were among the typical ones. Once a fact about our universe -- such as the existence of a thermodynamic arrow -- turns out to be typical, given the fundamental laws and the relevant boundary conditions, there is nothing left to explain (except, possibly, for the boundary conditions). On the flipside, atypical facts are usually the kind of facts that cry out for further explanation (cf. the contribution of Maudlin to this volume). And to accept the PH is precisely to assume that the initial state of our universe was atypical, relative to all possible microstates, in that it belonged to an extremely small (i.e., very low-entropy) macroregion. Penrose (1989) estimates the measure of this macroregion relative to the available phase space volume to be at most $1:10^{10^{123}}$ -- a mind-bogglingly small number. Notably, the explanatory pressure is mitigated by the fact that the PH entails only a special initial macrostate rather than a microscopic fine-tuning. In the case of a gas in a box, this would be the difference between atypical conditions that one can create with a piston and atypical conditions that one could create only by controlling the exact position and momentum of every single particle in the system. The point here is not that this makes it easier for a hypothetical creator of the universe, but that only the latter (microscopic) kind of fine-tuning gives rise to the worry that -- given the huge number of microscopic degrees of freedom and the sensitivity of the evolution to variations of the initial data -- atypical initial conditions could explain \emph{anything} (and thus explain nothing; cf. Lazarovici and Reichert (2015)). Nonetheless, the necessity of a PH implies that our universe looks very different from a typical model of the fundamental laws of nature -- and this is a fact that one can be legitimately worried about. 
\begin{figure}[ht]
	\begin{center}
	\includegraphics[width=0.6\textwidth]{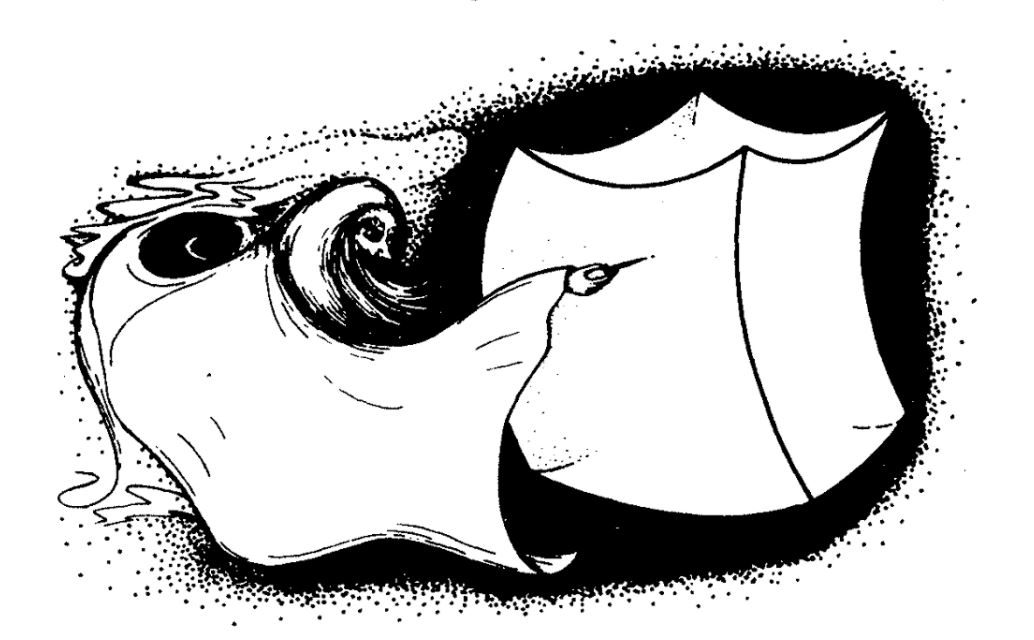}
	\end{center}
\caption{\footnotesize God picking out the special (low-entropy) initial conditions of our universe. Penrose (1989).}\label{fig:godinitialconditions}
\end{figure}

The second point of view was in particular defended by Callender (2004). While Callender is also sympathetic to the third option (regarding the PH as a law), he makes the broader case that a) there is no single feature of facts -- such as being atypical -- that makes them require explanation, and b) the conceivable explanations of the Past Hypothesis aren't much more satisfying than accepting it as a brute and basic fact. Notably, Ludwig Boltzmann himself eventually arrived at a similar conclusion: 
\begin{quote} The second law of thermodynamics can be proved from the mechanical theory if one assumes that the present state of the universe, or at least that part which surrounds us, started to evolve from an improbable state and is still in a relatively improbable state. This is a reasonable assumption to make, since it enables us to explain the facts of experience, and one should not expect to be able to deduce it from anything more fundamental. (Boltzmann, 1897)\end{quote}

The third option, finally, is most prominently advocated by Albert (2001) and Loewer (2007) in the context of the Humean \emph{Best System Account} of laws. Upon their view, the laws of nature consist in a) the microscopic dynamical laws b) the PH and c) a probability (or typicality) measure on the initial macroregion. This package has been dubbed the ``mentaculus'' (Loewer, 2012). It is supposed to correspond to the best-system-laws because it strikes the optimal balance between being simple and being informative about the history of our universe (the ``Humean mosaic''). In particular, adding b) and c) to the microscopic laws comes at relatively little cost in terms of simplicity but makes the system much more informative, precisely because it accounts for the thermodynamic arrow of time and allows for probabilistic inferences. In addition, Albert (2001) and Loewer (2007) employ the mentaculus in a sophisticated analysis of records,  counterfactuals, and more, the discussion of which goes beyond the scope of this paper. Instead, it is important to note that the proposition which Albert wants to grant the status of a law is not that the universe started in \emph{any} low-entropy state. The PH, in it's current form, is rather a placeholder for ``the macrocondition ... that the normal inferential procedures of cosmology will eventually present to us'' (Albert, 2001, p. 96). Ideally (we suppose), physics will one day provide us with a nice, and simple, and informative characterization of the initial boundary conditions of the universe -- maybe something along the lines of Penrose's Weyl Curvature Conjecture (Penrose, 1989) -- that would strike us as ``law-like''. But this is also what many advocates of option 1 seem to hope for as an ``explanation'' of the PH. So while option 3 sounds like the most clear-cut conclusion about the status of the PH, it is debatable to what extent it settles the issue. The more we have to rely on future physics to fill in the details, the less is already accomplished by calling the Past Hypothesis a law of nature.

\section{Thermodynamic arrow without a Past Hypothesis}\label{sec:SLwoPH}
In recent years, Sean Carroll together with Jennifer Chen (2004; see also Carroll (2010)), and Julian Barbour together with Tim Koslowski and Flavio Mercati (2013, 2014, 2015), independently put forward audacious proposals to explain the arrow of time \emph{without} a Past Hypothesis. While Barbour's arrow of time is not, strictly speaking, an \emph{entropic} arrow (but rather connected to a certain notion of complexity), Carroll's account is largely based on the Boltzmannian framework, although with a crucial twist. For this reason, we shall focus on the Carroll account first, before comparing it to the theory of Barbour et al. in Section \ref{sec:barbour}.

The crucial assumption of Carroll and Chen is that the relevant stationary measure on the phase space of the universe is unbounded, allowing for macrostates of \emph{arbitrarily high entropy} (while it's assumed that none has \emph{infinite} entropy). Hence, \emph{every} macrostate is a non-equilibrium state from which the entropy can typically increase in both time directions, defining a thermodynamic arrow -- or rather two opposite ones -- on either side of the entropy minimum. A typical entropy curve (one hopes) would thus be roughly parabolic or ``$\mathsf{U}$-shaped'', attaining its global minimum at some moment in time and growing monotonically (modulo small fluctuations) in both directions from this vertex (Fig. \ref{fig:carrollentropy}). Barbour et al. (2015) describe such a profile as ``one-past-two-futures'', the idea being that observers on each branch of the curve would identify the direction of the entropy minimum -- which the authors name the \emph{Janus point} -- as their past. In other words, we would have two future-eternal epochs making up the total history of the universe, with the respective arrows of time pointing in opposite directions. 
 
 \begin{figure}[ht]
 	\includegraphics[width=\textwidth]{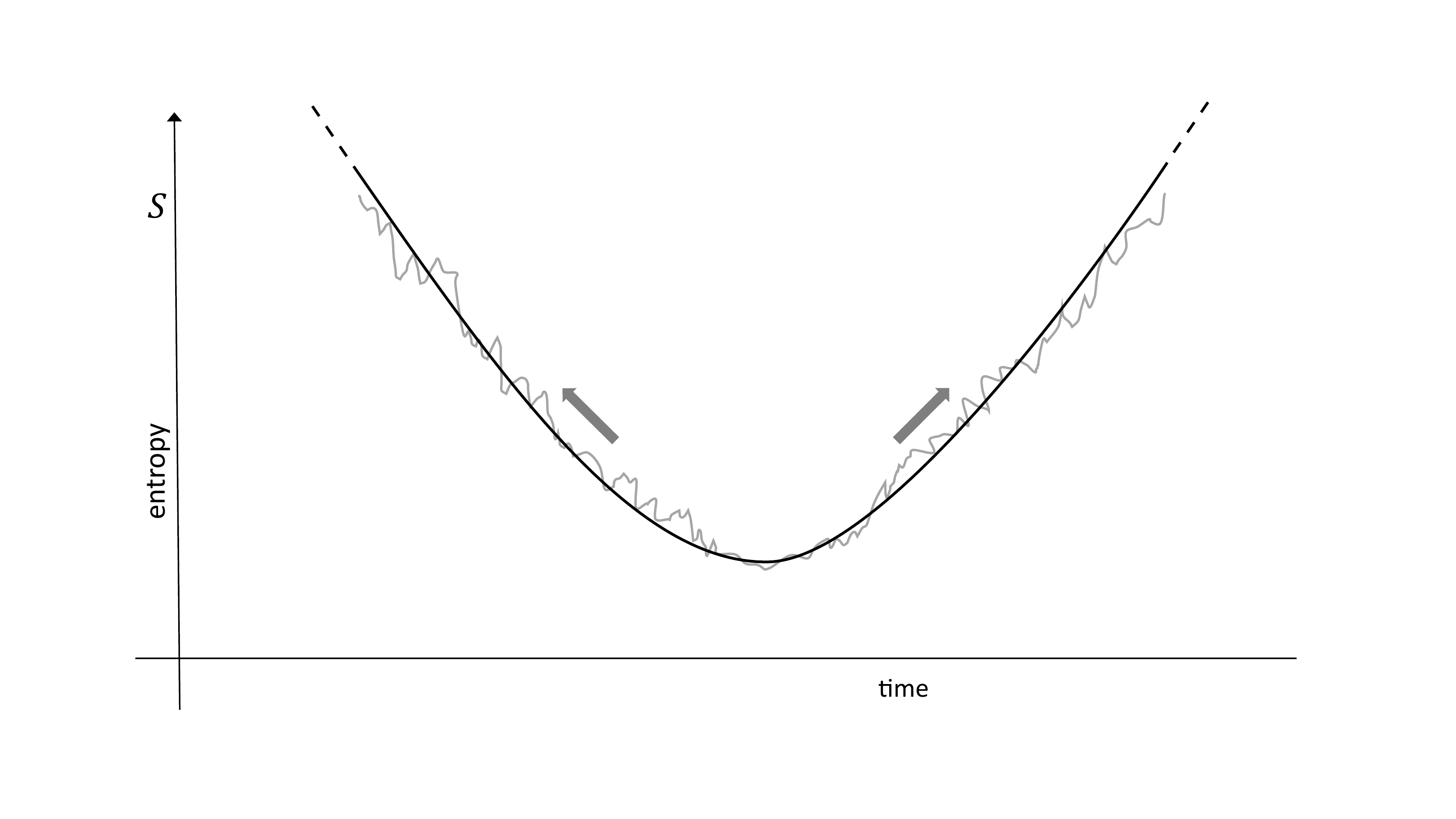}
 	\caption{\footnotesize Typical entropy curve (with fluctuations and interpolated) for a Carroll universe. The arrows indicate the arrow(s) of time on both sides of the ``Janus point'' (entropy minimum).} \label{fig:carrollentropy}
 \end{figure}

The Carroll model is intriguing because it is based on the bold, yet plausible assumption that the universe has no equilibrium state -- a crucial departure from the ``gas in the box'' paradigm that is still guiding most discussions about the thermodynamic history of the universe. And it is particularly intriguing for anybody worried about the status of the Past Hypothesis, because it seeks to establish the existence of a thermodynamic arrow in the universe as \emph{typical}. This is in notable contrast to the standard account, in which we saw that an entropy gradient is typical only under the assumption of atypical -- and time-asymmetric -- boundary conditions. 

Prima facie, it seems plausible that an eternal universe with unbounded entropy would exhibit the $\mathsf{U}$-shaped entropy profile shown in Fig. \ref{fig:carrollentropy}. For if we start in \emph{any} macrostate, the usual Boltzmannian arguments seem to suggest that typical microstates in the corresponding macroregion lead to a continuous increase of entropy in both time directions (since there are always vastly larger and larger macroregions, corresponding to higher and higher entropy values, that the microstate can evolve into). And then, any sensible regularization of the phase space measure would allow us to conclude that a $\mathsf{U}$-shaped entropy profile is typical \emph{tout court}, that is, with respect to all possible micro-histories. 

However, if we assume, with Carroll, a non-normalizable measure -- that assigns an infinite volume to the total phase space and thus allows for an unbounded entropy -- the details of the dynamics and the phase space partition must play a greater role than usual in the Boltzmannian account. For instance, the measure of low-entropy macroregions could sum up to arbitrarily (even infinitely) large values, exceeding those of the high-entropy regions. Or the high-entropy macroregions could be arbitrarily far away in phase space, so that the dynamics do not carry low-entropy configurations into high-entropy regions on relevant time-scales. The first important question we should ask is therefore: 
\begin{quote}
Are there any interesting and realistic dynamics that give rise to typical macro-histories as envisioned by Carroll and Chen?
\end{quote}
The original idea of Carroll and Chen (2004) is as fascinating as it is speculative. The authors propose a model of eternal spontaneous inflation in which new baby universes (or ``pocket universes'') are repeatedly branching off existing ones. The birth of a new universe would then increase the overall entropy of the multiverse, while the baby universes themselves, growing from a very specific pre-existing state (a fluctuation of the inflaton field in a patch of empty de-Sitter space), would typically start in an inflationary state that has much lower entropy than a standard big bang universe. This means, in particular, that our observed universe can look like a low-entropy universe, with an even lower-entropy beginning, even when the state of the multiverse as a whole is arbitrarily high up the entropy curve.  The details of this proposal are beyond the scope of our paper and do not (yet) include concrete dynamics or a precise definition of the entropy. 

In more recent talks, Carroll discusses a simple toy model -- essentially an ideal gas without a box -- in which a system of $N$ free particles can expand freely in empty space. The only macrovariable considered is the moment of inertia, $I=\sum\limits_{i=1}^N \mathbf{q}^2_i$, providing a measure for the expansion of the system. It is then easy to see that $I$ will attain a minimal value at some moment in time $t_0$, from which it grows to infinity in both time directions (cf. equation \eqref{lagrange-jacobi} below). The same will hold for the associated entropy since a macroregion, corresponding to a fixed value of $I$, is just a sphere of radius $\sqrt{I}$ in the position coordinates (while all momenta are constant). The entropy curve will thus have the suggested $\mathsf{U}$-shape with vertex at $t=t_0$. A detailed discussion of this toy model can be found in Reichert (2012), as well as Goldstein et al. (2016). 

In this paper, we will not discuss these two models in any more detail. Instead, we are going to argue in Section \ref{sec:graventropy} that there exists a dynamical theory fitting Carroll's entropy model that is much less speculative than baby universes and much more interesting, physically, than a freely expanding system of point particles. This theory is \emph{Newtonian gravity}. It will also allow us to draw interesting comparisons between the ideas of Carroll and Chen and those of Barbour, Koslowski, and Mercati. 

First, however, we want to address the question, whether this entropy model would even succeed in explaining away the Past Hypothesis. Are typical macro-histories as envisioned by Carroll and sketched in Fig. \ref{fig:carrollentropy} sufficient to ground sensible inferences about our past and future? Or would we still require -- if not the PH itself, then a close variant -- an equally problematic assumption about the specialness of the observed universe?


\section{The (dispensible) role of the Past Hypothesis}\label{sec:rolePH}

The question to be addressed in this section is thus the folllowing:

\begin{quote} 
	Can Carroll's entropy model ground sensible statistical inferences about the thermodynamic history of our universe without assuming (something akin to) a Past Hypothesis? 
\end{quote}

\noindent To approach this issue, and clarify the role of the PH in the standard account, we have to disentangle two questions that are often confounded in such discussions:
\begin{enumerate}[i)]
	
	\item  Given the fundamental laws of nature, what do typical macro-histories of the universe look like? In particular: is the existence of a thermodynamic arrow typical?
	
	\item Given our knowledge about the present state of the universe, what can we reasonably infer about its past and future?
	
\end{enumerate}

\noindent The answer to question i) will, in general, depend on the dynamical laws as well as cosmological considerations. If we have infinite time and a finite maximal entropy, a typical macro-history will be in thermodynamic equilibrium almost all the time, but also exhibit arbitrarily deep fluctuations into low-entropy states, leading to periods with a distinct entropy gradient, i.e., a local thermodynamic arrow. This \emph{fluctuation scenario} was, in fact, one of Boltzmann's attempts to resolve to the hard problem of irreversibility (Boltzmann, 1896). 

However, to assume a fluctuation as the origin of our thermodynamic arrow is highly unsatisfying, Feynman (1967, p. 115) even calls it ``ridiculous''. The reason is that fluctuations which are just deep enough to account for our present macrostate are much more likely (and would thus occur much more frequently\footnote{e.g. in the sense $\limsup\limits_{T \to +\infty} \frac{1}{T} \bigl(\,  \# \text{fluctuations with entropy minimum} \approx S \text{ in the time-interval } [-T, T] \, \bigr)$}) than fluctuations producing an even lower-entropy past from which the current state could have evolved in accordance with the second law. We would thus have to conclude that we are currently experiencing the local entropy minimum, that our present state -- including all our records and memories -- is in fact the product of a random fluctuation rather than a lower-entropy past. Feynman makes the further case that the fluctuation scenario leads not only to absurd conclusions about the past, but to wrong ones about the present state of the universe, as it compels us to assume that our current fluctuation is not any deeper -- and hence more unlikely -- than necessary to explain the evidence we already have: If we dig in the ground and find a dinosaur bone, we should not expect to find other bones nearby. If we stumble upon a book about Napoleon, we should not expect to find other books containing the same information about a guy called Napoleon. The most extreme form of this reductio ad absurdum is the \emph{Boltzmann brain problem} (see, e.g., Carroll (2010) for a recent discussion): a fluctuation that is just deep enough to account for your empirical evidence (many people claim) would produce only your brain, floating in space, with the rest of the universe at equilibrium. You should thus conclude that this is, by far, the most likely state of the universe you currently experience. 


The only possible escape in such a fluctuation scenario is to invoke the additional postulate -- a form of Past Hypothesis -- that the present macrostate is not the bottom of the fluctuation, but has been preceded by a sufficiently long period of entropy increase from a state of much lower entropy, still. In this context, the PH would thus serve a \emph{self-locating} function, taking the form of an indexical proposition that locates our present state on the upwards-slope of a particularly deep fluctuation (Fig. \ref{fig:selflocating}). 

\begin{figure}[ht]
	\includegraphics[width = \textwidth]{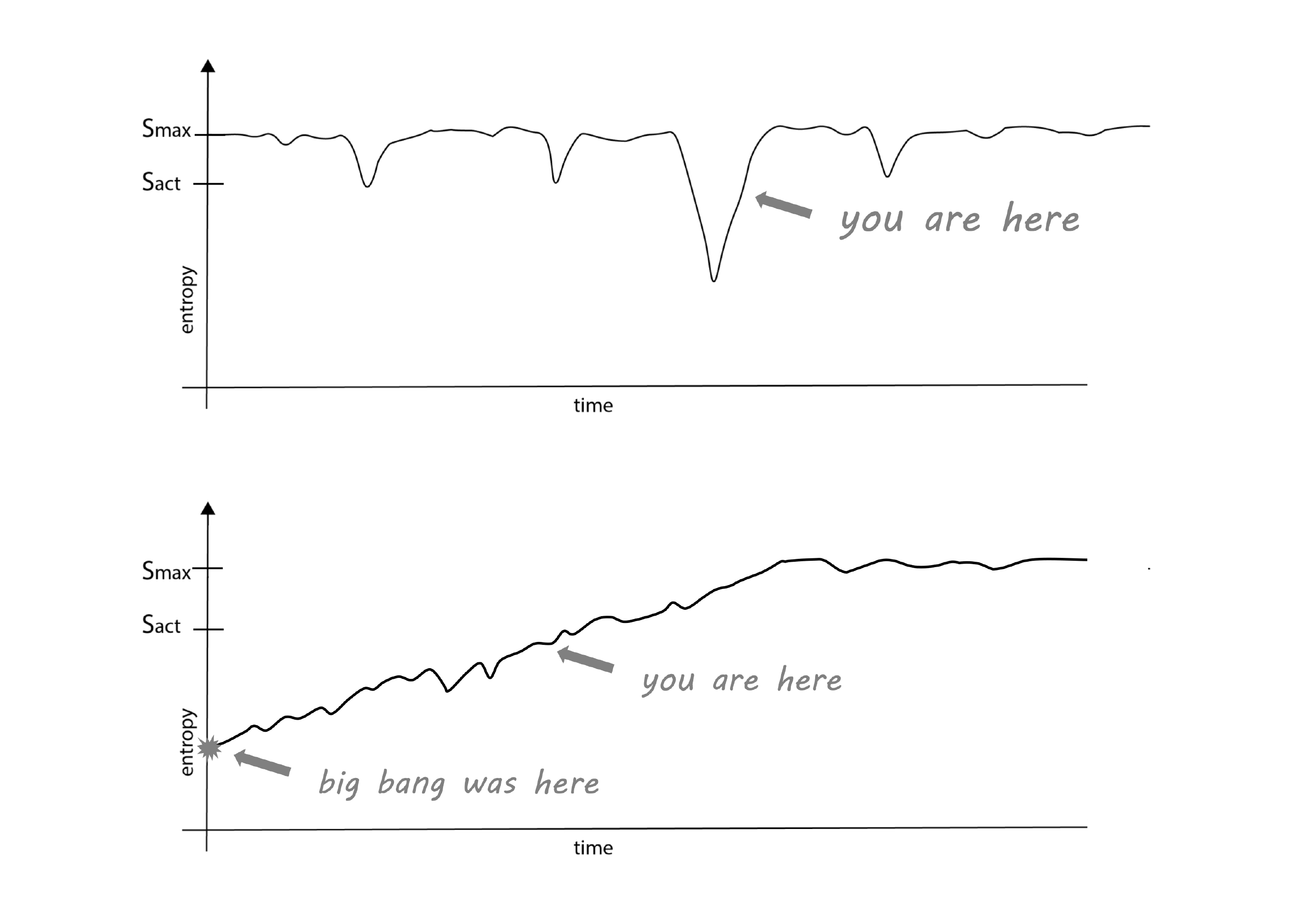}
	\caption{\footnotesize Self-location hypothesis in the fluctuation scenario (upper image) and big bang scenario (lower image) with bounded entropy. Time-scales in the upper image are much larger than below and periods of equilibrium are much longer than depicted.} \label{fig:selflocating}
\end{figure}

The now standard account assumes a bounded entropy and a relatively young universe  -- about 13.8 billion years according to estimates from modern big bang cosmology. In this setting (we interpret the big bang as the actual beginning of time), a typical history would not have any thermodynamic arrow at all (the time-scale of $\sim 10^{10}$ years is too short for significant entropy fluctuations on cosmological scales). Thus, we need the PH to account for the existence of a thermodynamic arrow in the first place by postulating a low-entropy boundary condition at the big bang. A self-locating proposition is still crucial and hidden in the assumption of a young universe. Winsberg (2012) makes it explicit in what he calls the ``Near Past Hypothesis'' (NPH), which is that our present state lies between the low-entropy beginning of the universe and the time of first relaxation into equilibrium.  Without such an assumption -- and assuming that the universe is eternal in the future time direction -- we would essentially be back in a fluctuation scenario with all its Boltzmann-brain-like absurdities. In a future-eternal universe with bounded entropy, there are still arbitrarily many entropy fluctuations that are just deep enough to account for our present evidence (but not much deeper). And we would still have to conclude that we are much more likely in one of these fluctuations than on the initial upwards slope originating in the very low-entropy big bang (cf. also Loewer, forthcoming). 

The self-locating role of the PH (which we take to include the NPH -- for what would be the point otherwise?) is thus indispensable. And it is, in fact, the indexical proposition involved, rather than the non-dynamical boundary condition, that we would be surprised to find among the fundamental laws of nature as we consider this option for the status of the Past Hypothesis.

Carroll's model, finally, postulates an eternal universe and unbounded entropy, suggesting that typical macro-histories will have the $\mathsf{U}$-shaped entropy profile depicted in Fig. \ref{fig:carrollentropy}. If this works out -- and we will argue that it does, but at least see no reason why it couldn't -- the existence of a thermodynamic arrow (respectively two opposite ones) will be \emph{typical}. (For completeness, we could also discuss the option of a temporally finite universe and unbounded entropy, but this model does not seem advantageous and goes beyond the scope of the paper.) In the upshot, the Carroll model can indeed explain the existence of a thermodynamic arrow without invoking a PH as a fundamental postulate over and above the microscopic laws and the pertinent typicality measure. 
It may still turn out that the theory requires a PH for its self-locating function, \emph{if} it would otherwise imply that our current macrostate should be the global entropy minimum, i.e., has not evolved from a lower-entropy past. The relevant version of the PH may then take the form of an indexical clause -- stating that our present state is high up the entropy curve -- or be a characterization of the entropy minimum (Janus point) of our universe. (In the first case, the PH would first and foremost locate the present moment within the history of an eternal universe, in the latter, it would first and foremost locate the actual universe within the set of possible ones.) But it is not obvious why the Carroll model would lead to the conclusion that we are currently at (or near) the entropy minimum, and the issue actually belongs to our second question -- how to make inferences about the past -- to which we shall now turn.

\subsection{Predictions and retrodictions}

The most straightforward response to question ii) -- how to make inferences about the past or future -- is the following method of statistical reasoning: Observe the current state of the universe (respectively a suitably isolated subsystem), restrict the pertinent probability (more correctly: typicality) measure to the corresponding macroregion in phase space, and use the conditional measure to make probabilistic inferences about the history of the system. We shall call this \emph{naive evidential reasoning} (reviving a terminology introduced in an unpublished 2011 draft of Goldstein et al., 2016). The negative connotation is warranted because we know that while this kind of reasoning works well for \emph{predictions} -- inferences about the future -- it leads to absurd, if not self-refuting, conclusions when applied for \emph{retrodictions} -- i.e., inferences about the past. 

The now standard move to avoid this predicament is to employ the PH to block naive evidential reasoning in the time direction of the low-entropy boundary condition. For sensible retrodictions, we learn, one must conditionalize on the low-entropy initial state in addition to the observed present state. It is rarely, if ever, noted that an appeal to a PH may be sufficient but not necessary at this point. The key is to appreciate that the second question -- how to make inferences about the past and future of the universe -- must be addressed subsequently to the first -- whether a thermodynamic arrow in the universe is typical. For if we have good reasons to believe that we live in a universe with a thermodynamic arrow of time, this fact alone is sufficient to conclude the irrationality of retrodicting by conditionalizing the phase space measure on the present macrostate. More precisely, it follows from the Boltzmannian analysis that in a system with a thermodynamic arrow, the evolution towards the future (the direction of entropy increase) looks like a \emph{typical} one relative to any intermediate macrostate, while the actual microstate is necessarily atypical with respect to its evolution towards the entropic past (see Fig. \ref{fig:pastfutureMS}). This is essentially the reversal of the familiar paradox that entropy increase in \emph{both} time directions comes out as typical relative to any non-equilibrium macrostate. 

In the upshot, the fact that naive evidential reasoning doesn't work towards the entropic past can be inferred from the existence of a thermodynamic arrow; it does not have to be inferred from the assumption of a special initial state. The explanation of the thermodynamic arrow, in turn, may or may not require a special initial state, but this was a different issue -- discussed above. 

\begin{figure}[ht]
	\centering
	\includegraphics[width = 0.9 \textwidth]{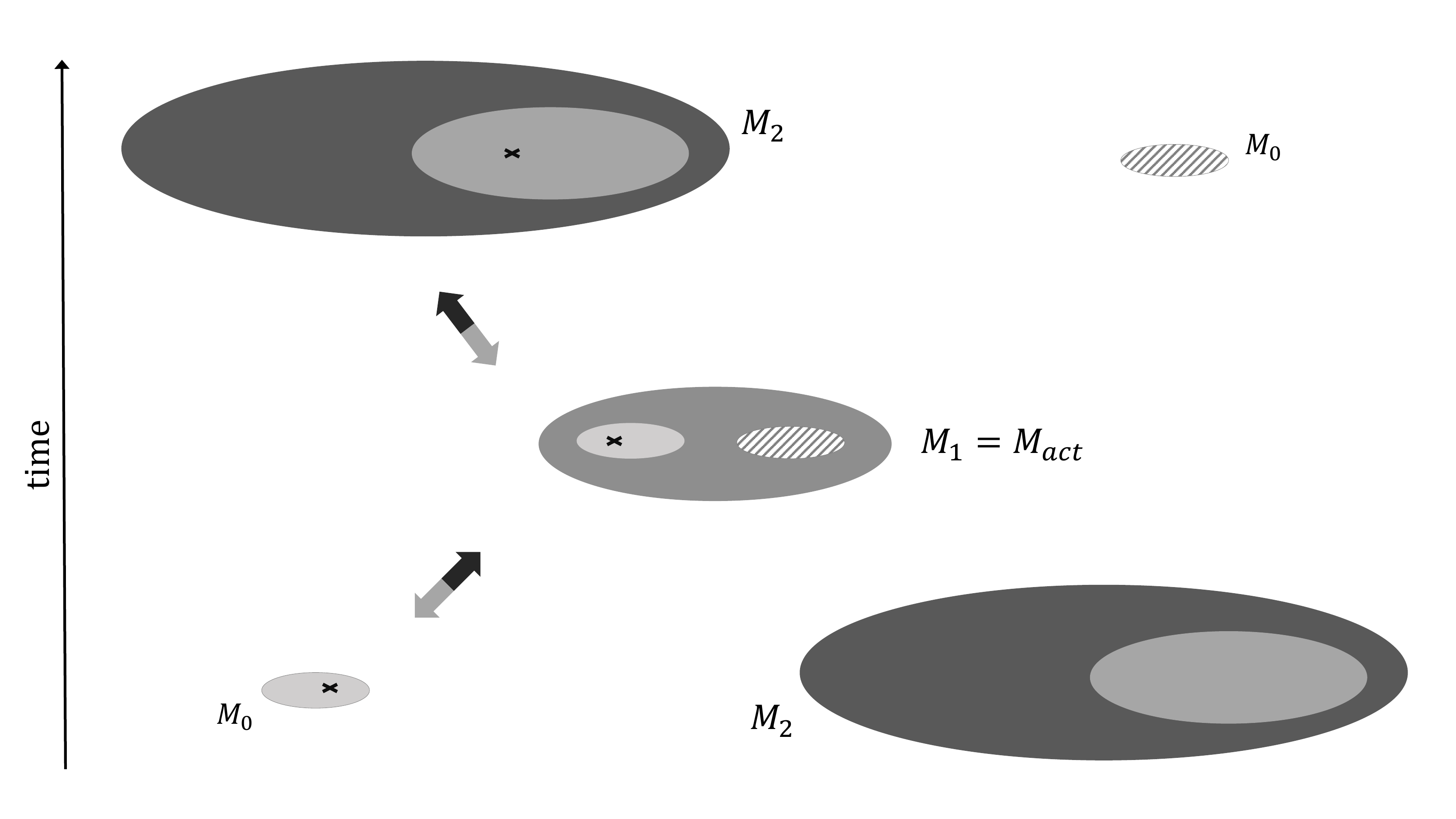}
	\caption{\footnotesize We assume, for simplicity, that all macrostates are invariant under the time-reversal transformation ($(q,p) \to (q,-p)$ in classical mechanics). Typical microstates in the intermediate macro-region $M_1=M_{act}$ evolve into a higher-entropy region $M_2$ in both time directions. Only a small subset of microstates (light grey area) have evolved from the lower-entropy state $M_0$ in the past; an equally small subset (shaded area) will evolve into $M_0$ in the future. The actual microstate (cross) has evolved from the lower-entropy state in the past; only its future time-evolution corresponds to the typical one relative to the macrostate $M_{act}$.}  \label{fig:pastfutureMS}  
	
\end{figure}

If the relevant physical theory tells us that a thermodynamic arrow is typical, i.e., exists in almost all possible universes, we have a very strong theoretical justification for believing that we actually live in a universe with a thermodynamic arrow. And if we believe that we live in a universe with a thermodynamic arrow,  a rational method for making inferences about the past is not naive evidential reasoning but \emph{inference to the best explanation}. Rather than asking ``what past state (or states) would be typical given our present macrostate?'', we should ask ``what past state (or states) would make our present macrostate typical?''. In more technical terms, rather than adjusting our credence that the past macrostate was $M_0$ according to $\IP(M_{0} \mid M_{act})$, where $M_{act}$ is our present macrostate,  we should ``bet'' on macrostates $M_0$ that maximize  $\IP(M_{act} \mid M_{0})$. If we find a dinosaur bone, we should infer a past state containing a dinosaur. If we find history books with information about Napoleon, we should infer a past state containing a French emperor by the name of Napoleon. We do not claim that this amounts to a full-fledged analysis of the epistemic asymmetry, but it is a large enough part, at least, to uphold the reliability of records and get a more fine-grained analysis going. In particular, considering the universe as a whole, the fact that it has evolved from a lower-entropy state in the past is \emph{inferred}, rather than assumed, by this kind of abductive reasoning.  

By now it should be clear that the debate is not about whether some version of the PH is \emph{true}, but about whether it is an \emph{axiom}. And the upshot of our discussion is that if the existence of a thermodynamic arrow in the universe turns out to be typical, we can consider our knowledge of the low-entropy past to be reasonably grounded in empirical evidence and our best theory of the microscopic dynamics (as any knowledge about our place in the history of the universe arguably should).

Another way to phrase the above analysis goes as follows: Naive evidential reasoning applied to both time directions will always lead to the conclusion that the current macrostate is the (local) entropy minimum. However, if we know that we observe a universe (or any other system) with a thermodynamic arrow, we also know that this conclusion would be wrong \emph{almost all the time}. More precisely, it would be wrong unless we happened to observe \emph{a very special period} in the history of the universe in which it is close to its entropy minimum. 

Goldstein, Tumulka, and Zangh\`i (2016) provide a mathematical analysis of this issue in the context of Carroll's toy model of freely expanding particles. Their discussion shows that the two opposing ways of reasoning -- typical microstates within a given macroregion versus typical time-periods in a history characterized by a $\mathsf{U}$-shaped entropy curve -- come down to different ways of regularizing the unbounded phase space measure by choosing an appropriate cut-off. Goldstein et al. then argue against the first option, corresponding to naive evidential reasoning, by saying that certain facts about the past amount to ``pre-theoretical'' knowledge. We provided a concurrent argument based on a theoretical (Boltzmannian) analysis. Nonetheless, from a formal point of view, a certain ambiguity remains. In Section \ref{sec:barbour}, we will discuss how the relational framework of Barbour et al. is able to improve upon this situation.  

\subsection{The mystery of our low-entropy universe}

Another possible objection to the Carroll model (disregarding baby universes) goes as follows: Doesn't the fact that the entropy of the universe could be arbitrarily high make its present very low value -- and the even lower value at the Janus point -- only more mysterious? In other words: doesn't the fact that the entropy could have been arbitrarily high just increase the explanatory pressure to account for the specialness of the observed universe? While we cannot completely deny the legitimacy of this worry, our intuition is that the Carroll model precludes any \emph{a priori} expectation of what the entropy of the universe \emph{should} be. If it can be arbitrarily (but not infinitely) large, any possible value could be considered ``mysteriously low'' by skeptics. 

Sidney Morgenbesser famously responded to the ontological question \emph{Why is there something rather than nothing}: ``If there was nothing, you'd be still complaining!'' In the same spirit (though not quite as witty), our reaction to the question \emph{Why is the entropy of the universe so low?} would be: ``If it was any higher, you'd be still complaining!'' 

We concede, nonetheless, that divergent intuitions about this question are possible. In fact, the ambiguity is once again paralleled by mathematical issues arising from the non-normalizability of the phase space measure. When Penrose (1989) estimates that the entropy of the universe near the big bang could have been about $10^{123}$ times higher, the common worry is not that the actual value was so low in comparison, but that a $10^{123}$ times higher entropy would seem $10^{10^{123}}$ times \emph{as likely}. While this conclusion is questionable even in the standard Boltzmannian framework (with a finite phase space volume), the interpretation of a non-normalizable phase space measure as a \emph{probability} measure is problematic, to say the least, leading in particular to the paradox that any finite range of entropy values has probability zero. Again, we'll have to leave it at that as far as the discussion of the Carroll model is concerned, explaining instead in the last section how the shape space theory of Barbour et al. is able to resolve the issue. First however, we owe the reader some evidence that we haven't been discussing the empty set, but that Newtonian gravity might provide a relevant example of a Carroll universe.


\section{Entropy of a classical gravitating system}\label{sec:graventropy}
There is a lot of confusion and controversy about the statistical mechanics of classical gravitating systems, despite the fact that statistical methods are commonly and successfully used in areas of astrophysics that are essentially dealing with the Newtonian $N$-body problem (see, e.g., Heggie and Hut (2003)). (An excellent paper clearing up much of the confusion is Wallace (2010); see Callender (2009) for some problematic aspects of the statistical mechanics of gravitating systems and Padmanabhan (1990) for a mathematical treatment.) Some examples of common claims are: 

\begin{enumerate}[a)]
	\item Boltzmann`s statistical mechanics is not applicable to systems in which gravity is the dominant force. 
	\item The Boltzmann entropy of a classical gravitating system is ill-defined or infinite. 
	\item An entropy increasing evolution for a gravitating system is exactly opposite to that of an ideal gas. While the tendency of the latter is to expand into a uniform configuration, the tendency of the former is to clump into one big cluster. 
		
\end{enumerate}

We believe that the first two propositions are simply false, while the third is at least misleading. However, rather than arguing against these claims in the abstract, we will provide a demonstration to the contrary by proposing an analysis of a classical gravitating system in the framework of Boltzmann's statistical mechanics.

We start by looking at the naive calculation, along the lines of the standard textbook computation for an ideal gas, that finds the Boltzmann entropy of the classical gravitating system to be infinite (see, e.g., Kiessling (2001)). For simplicity, we always assume $N$ particles of equal mass $m$. We have
\begin{align}
S(E,N,V) := k_B \log \lvert\Gamma(E,N,V) \rvert = k_B \log \Bigl[\frac{1}{h^{3N}N!} \int\limits_{V^N} \int\limits_{\IR^{3N}} \delta\bigl(H - E \bigr)\ \mathrm{d}^{3N}q \, \mathrm{d}^{3N}p\Bigl],
\end{align}
with 
\begin{align} H^{(N)}(q,p) = \sum\limits_{i=1}^N \frac{\mathbf{p}^2_i}{2m} - \sum\limits_{1\leq i < j \leq N} \frac{Gm^2}{\lvert \mathbf{q}_i - \mathbf{q}_j\rvert}\end{align}
\noindent and 
\begin{align} \label{eq:divergence}
\int\limits_{V^N} \int\limits_{\IR^{3N}} \delta\bigl(H - E \bigr)\ \mathrm{d}^{3N}p \; \mathrm{d}^{3N}q
=  C \, \int\limits_{V^N} \biggl(E +  \sum\limits_{ i < j} \frac{Gm^2}{\lvert \mathbf{q}_i - \mathbf{q}_j\rvert} \biggr)^{\frac{3N-2}{2}} \mathrm{d}^{3N}q = +\infty.
\end{align}
For $N>2$, the integral \eqref{eq:divergence} diverges due to the singularity of the gravitational potential at the origin. Note that $V$ need not correspond to the volume of a physical ``box'' confining the particles; in empty space, it would merely describe a volume (e.g. the smallest sphere, or cuboid, or convex set) enclosing the particle configuration. 

There is nothing mathematically wrong with the above calculation, it just doesn't actually compute what it's supposed to. One problem is that as we integrate over $V^N$, we sum over all possible configurations of $N$ particles (with total energy $E$) within the volume $V$. This includes configurations in which the particles are homogeneously distributed, but also configurations in which most particles are concentrated in a small subset of $V$ (Fig. \ref{fig:volumeintegral}). In the case of the ideal gas in a box, the contribution of the latter is negligible since almost the entire phase space volume is concentrated on spatially homogeneous configurations. It is the entropy (or phase space volume) of this equilibrium state that we actually want to compute, and the mistake we make by including non-equilibrium configurations (in which the particles are concentrated in one half, or one quarter or one third, etc. of the volume) is so small that it is hardly ever mentioned. 

\begin{figure}[ht]
	\includegraphics[width = \textwidth]{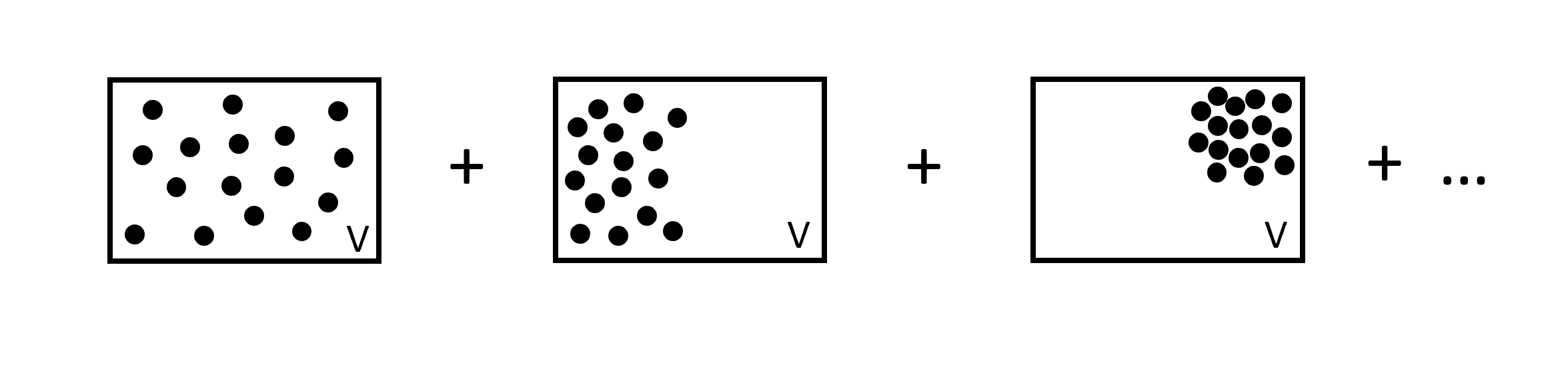}
	\caption{\footnotesize Performing the volume integral in \eqref{eq:divergence}, we sum over \emph{all} possible configurations of the particles within the given volume $V$.}\label{fig:volumeintegral}
\end{figure}
In the case of a gravitating system, the situation is distinctly different, since the spatial configuration is correlated with the kinetic energy or, in other words, with the possible momentum configurations of the system. Simply put, for an attractive potential and constant energy, a more concentrated spatial configuration corresponds to higher kinetic energy and thus larger phase space volume in the momentum variables. The ``total volume'' $V$ is thus not a good macrovariable to describe a system with gravity. In particular, if we want to know whether the entropy of a gravitating system is increasing as the configuration clusters, we have to consider a macroscopic variable that actually distinguishes between more and less clustered configurations (as those shown in Fig. \ref{fig:volumeintegral}). 

We thus propose to describe a system of $N$ gravitating point particles by the following set of macrovariables:

\begin{itemize}
\item $E(p, q) = \frac{p^2}{2m} + V(q) = \sum\limits_{i=1}^N \frac{\mathbf{p}^2_i}{2m} - \sum\limits_{1\leq i<j\leq N} \frac{Gm^2}{|\mathbf{q}_i-\mathbf{q}_j|}$  is the total energy of the system which is a constant of motion.

\item $I(q) = \sum\limits_{i=1}^N m (\textbf{q}_i-\sum\limits_{j=1}^N \textbf{q}_j)^2$ is the moment of inertia that will quantify how much the particles are spread out over space. In the center of mass frame -- which we can and will use without loss of generality -- it simplifies to $I(q)=mq^2= \sum_{i=1}^N m \textbf{q}_i^2$. Notably, our system is not confined by physical boundaries but can expand arbitrarily in empty space. $I(q)$ can thus grow without bound.  
\end{itemize}

\noindent The moment of inertia alone is still too coarse to differentiate between, let's say, a uniform configuration and a concentrated cluster with few residual particles far away. To distinguish between more and less clustered configurations, we thus have to introduce a further macrovariable. We choose: 
\begin{itemize}
	\item $U(q):= -V(q)= \sum\limits_{1\leq i<j\leq N}  \frac{Gm^2}{|\textbf{q}_i-\textbf{q}_j|}$, which is just the absolute value of the potential energy. Since the total energy is $E(q,p)=T(p) + V(q)$, specifying the value of $E$ and $U$ is equivalent to specifying $E$ and the kinetic energy $T(p)= \sum\limits_{i=1}^N \frac{\mathbf{p}^2_i}{2m}$. An increase of $U(q)$ thus signifies both clustering and heating of the system. 
	
	Note that defining macrostates in terms of $U$ (respectively the potential energy) automatically takes care of the \emph{ultraviolet divergence} in the computation of the associated entropy since the minimal particle distance $r$ is bounded as $r\geq \frac{Gm^2}{U}$. 
\end{itemize}
Obviously, we do not claim that the moment of inertia or the gravitational poential energy of the universe can be precisely measured. What makes them macrovariables is, first and foremost, the fact that they are coarse-graining: many different micro-configurations of an $N$-particle universe realize the same values of $I$ and $U$. Moreover, in the next subsection, it will become clearer that the evolution of these macrovariables does indeed provide relevant information about the large-scale structure of a gravitating universe. 

\noindent Now, to determine the entropy of the respective macrostates, we have to compute the phase space volume corresponding to a macroregion $\Gamma(E, I, U)$, that is

\begin{align}\lvert \Gamma(E, I, U) \rvert = \iint\limits_{\IR^{3N}\times\IR^{3N}} \delta\Bigl(\frac{p^2}{2m} + V(q) - E \Bigr)\; \delta\Bigl(V(q) + U\Bigr)\;\delta\Bigl(m q^2 - I \Bigr)\,  \mathrm{d}^{3N}q \, \mathrm{d}^{3N}p
\end{align}
for fixed values of $E, I,$ and $U$. Unfortunately, we weren't able to solve this integral analytically (and maybe this is, in fact, impossible). However, if we replace the sharp values of the macrovariables with a small interval $I(q) \in (I -\Delta I,I +\Delta I), \, \lvert V(q)\rvert \in  (U-\Delta U, U + \Delta U)$ with e.g. $\Delta I = \frac{I}{\sqrt{N}}, \Delta U = \frac{U}{\sqrt{N}}$ (roughly a standard deviation), we can obtain the bounds: 
\begin{equation*}  C\, \Bigl(\frac{G m^{5/2}}{\sqrt{I}U}\Bigr)^3 e^{-5N} (E+U)^{\frac{3N-2}{2}}  {I}^\frac{3N}{2} \leq|\Gamma(E, U\pm \Delta U, I\pm \Delta I)| \leq C e^{\sqrt{N}} (E+U)^{\frac{3N-2}{2}} {I}^\frac{3N}{2} , \end{equation*}
for sufficiently large values of $I$ and $U$ (more precisely, of the dimensionless quantity $\frac{\sqrt{I}U}{G m^{5/2}}$) and $E\geq 0$, where $C$ is a positive constant depending only on $N$ and $m$. A precise statement and proof (valid for any $E$) is given in the appendix. Thus, we have 
\begin{equation}\label{phasespacevolume}
|\Gamma(E, U, I)| \approx {const.}\cdot \bigl(I(E+U))^\frac{3N}{2},
\end{equation} and, ignoring an additive constant,
\begin{equation}\label{graventropy} S(E,I,U) \approx \frac{3N}{2} \Bigl(\log(E+U) + \log(I)\Bigr). \end{equation}

\subsection{Typical evolutions} \label{sec:typicalevol}
We now provide a discussion of this result. 

\begin{enumerate}
	\item With our choice of macrovariables, the associated Boltzmann entropy of a gravitating system is well-defined and finite. We also see that the entropy can grow without bounds, either due to continuous expansion of the system ($I \to + \infty$) and/or due to continuous clustering and self-heating ($U \to +\infty$). 
	\item While common wisdom says that the typical evolution of a gravitating system is one of clumping and clustering, our computation shows that clustering and expansion (as quantified by the macrovariable $U$ and $I$, respectively) can contribute equally to an increase of entropy. This fits well with the observed processes of gravithermal collapse that are known to show a ``core-halo'' pattern (see, e.g., Heggie and Hut (2003, Ch. 23)): the configuration of masses splits into a core that collapses and heats up (increase of $U$) and a collection of particles on the outskirts that are blown away (increase of $I$). 
	
	On even larger (cosmological) scales, a gravitating system in a homogeneous configuration can increase its entropy along both ``dimensions'' by forming many local clusters (``galaxies'') that disperse away from each other -- a process that would look very much like structure formation!
	
	Hence, it seems to be precisely the interplay between the opposing tendencies of clustering and expansion that makes classical gravity much more interesting, from a thermodynamic point of view, than often assumed.
	
	\item  Analytical and numerical results support the conclusion that the typical evolution of a gravitating system is one in which the entropy \eqref{graventropy} increases from a minimum value in both time directions, giving rise to the $\mathsf{U}$-shaped entropy curves proposed by Carroll and Chen.  
The first analytical result is the classical \emph{Lagrange-Jacobi equation} for the gravitational potential: 
		\begin{equation}\label{lagrange-jacobi} \ddot{I} = 4E-2V. \end{equation}
From this equation, which is a standard result in analytical mechanics, it follows immediately that if $E\ge 0$, the second time derivative of the moment of inertia is strictly positive (note that $V$ is negative), meaning that $I(t)$ is a strictly convex (upwards curving) function. Together with the fact that $I\to \infty$ as $t \to \pm \infty$ (Pollard, 1967), we can conclude that the graph of $I$ has precisely the kind of $\mathsf{U}$-shape that we expect for the entropy.

	Thanks to the results of Saari (1971) and Marchal and Saari (1976), we have an even more precise picture of the asymptotic behavior of the Newtonian gravitational $N$-particle system. Their work studies the inter-particle distances $|\textbf{q}_i-\textbf{q}_j|$, as well as the dispersion from the center of mass, for $t \to \infty$ independent of the total energy. It is found that either the minimal particle distance goes to zero \begin{equation*}
		\lim\limits_{t\to\infty} r(t):=\lim\limits_{t\to\infty} \min_{i\neq j} |\textbf{q}_i(t)-\textbf{q}_j(t)|=0,
		\end{equation*} while the greatest particle distance goes to infinity faster than $t$ \begin{equation*}\lim\limits_{t\to\infty}\frac{R(t)}{t} :=\lim\limits_{t\to\infty}\,  t^{-1} \max_{i\neq j} |\textbf{q}_i(t)-\textbf{q}_j(t)| = \infty, \end{equation*}
or the asymptotic behavior in the center-of-mass frame is characterized by 
	\begin{equation}\label{asymptotic}\textbf{q}_i(t)=\textbf{A}_i t + \mathcal{O} (t^{2/3})\;\; \forall i=1,\ldots, N \hspace{0.5cm} \mathrm{and} \hspace{0.5cm} \limsup\limits_{t \to \infty} r >0, \end{equation}
	where $\textbf{A}_i \in \IR^3$ are constant vectors (possibly the zero vector). Note that since the dynamics are time-reversal invariant, the results hold for $t\to -\infty$, as well.

The first case describes so-called ``super-hyperbolic escape''. This scenario is consistent with an increase of our gravitational entropy \eqref{graventropy}, implying both $I \to \infty$ and $U \to \infty$ as $t \to \infty$. It also includes, however, the pathological cases in which solutions diverge in finite time. It is the second case (when super-hyperbolic escape is excluded), in which the Newtonian $N$-body system is much more interesting and generally well-behaved. More precisely, we see that if  \eqref{asymptotic} holds, all inter-particle distances fall into one of the following three classes (see Saari (1971), Cor. 1.1, together with Marchal and Saari (1976), Cor. 6):
\begin{equation} |\textbf{q}_i-\textbf{q}_j|= L_{ij}t  +  \mathcal{O}(t^{2/3}),  \end{equation}
or  
\begin{equation} |\textbf{q}_i-\textbf{q}_j| = \mathcal{O}(t^{2/3}) \end{equation}
or  
\begin{equation} |\textbf{q}_i-\textbf{q}_j| \leq L\end{equation}
for large $t$ and positive constants $L, {L}_{ij}$. 

	The result can be summarized as follows (cf. Saari (1971), p. 227): On sufficiently large time-scales, the system forms clusters, consisting of particles whose mutual distances remain bounded. These clusters form subsystems (clusters of clusters) that are reasonably well isolated (energy and angular momentum are asymptotically conserved in each one of them separately), the distance between their centers of mass growing proportional to $t$. Finally, within each of these subsystems, the clusters separate approximately as $t^{2/3}$. In other words, the long-term behavior of such a Newtonian universe looks very much like \emph{structure formation}, with local clumping into ``galaxies'' and  global expansion due to galaxies and galaxy clusters receding from each other. 
	
	In regard to entropic considerations,  i.e., equation \eqref{graventropy}, we note that the moment of inertia will grow asymptotically like $I(t) \sim t^2$, while the macrovariable $U(t)$ is at least bounded from below by some multiple of $\frac{N}{L^2}$ (assuming that the number of particles in a cluster is of order $N$). What happens at intermediate times? Assuming henceforth a non-negative total energy, we already know that $I(t)$ is strictly convex. Together with its quadratic growth for $t \to \pm \infty$, we can conclude that it has a unique global minimum, let's say at $t=\tau$, from which it increases in both time directions. $U(t)$ will in general fluctuate, but if we exclude particle collisions and ``near particle collisions'' (very close encounters), it will remain bounded and not fluctuate too quickly ($\dot{U}$ remains bounded, as well). Hence, we expect that the graph of $(E + U(t)) I(t)$ (the logarithm of which is proportional to our gravitational entropy) looks qualitatively like that of $I(t)$, namely by and large parabolic. Indeed, numerical simulations by Barbour et al. (2013, 2015) for the $E=0$ universe (with $N=1000$ and random initial data) support the claim that the evolution of $I\cdot U$ is well interpolated by a parabola of the form $\alpha (t-\tau)^2 + \beta$ with $\alpha, \beta >0$. All this suggests the desired $\mathsf{U}$-shaped evolution of the entropy $S(E,I,U) \approx \frac{3N}{2} \Bigl(\log(E+U) + \log(I)\Bigr)$ as a function of time for a Newtonian gravitating universe with non-negative energy. (Actually, on large time-scales, the shape looks less like a $\mathsf{U}$ and more like $\mathsf{\curlyvee}$ -- how some children draw birds on the horizon -- since $S(t)$ grows only logarithmically as $\lvert t-\tau\rvert \to \infty$.)

\end{enumerate}
We conclude that a Newtonian gravitating universe is indeed a ``Carroll universe'' which has no equilibrium state and for which entropy increase (in opposite directions from a global minimum) is typical. This entropy increase is, moreover, consistent with structure formation. It does not merely lead to one big boring clump of matter.

\section{Gravity and typicality from a relational point of view}\label{sec:barbour}

Starting from Machian / Leibnizian principles, Barbour, Koslowski, and Mercati (2013, 2014, 2015) discuss the Newtonian gravitational system from a relationalist perspective. According to the relational framework that Julian Barbour has championed over the past decades, all physical degrees of freedom are described on \emph{shape space} $S$ which is obtained from Newtonian configuration space by factoring out global rotations, translations, and scale, leaving us with a $3N-7$ dimensional space for an $N$-particle system. The configuration of $N$ point particles is then characterized by the angles and ratios between their (Euclidean) distance vectors -- or, in other words, by its \emph{shape} -- independent of extrinsic scales. The lowest-dimensional (non-trivial) shape space is that of $N=3$ particles. In this case, the shapes are those of triangles -- specified by 2 angles or the ratios between 3 distances -- and the topology of shape space that of a 2-dimensional (projective) sphere. 

Considering standard Newtonian gravity on absolute space and trying to extract, so to speak, its relational essence, we have to eliminate all dependencies on extrinsic spatio-temporal structures. To this end, we restrict ourselves to models with vanishing total momentum, $\textbf{P}=\sum_{i=1}^N \textbf{p}_i \equiv 0$, and angular momentum, $\textbf{L}=\sum_{i=1}^N \textbf{q}_i\times\textbf{p}_i \equiv 0$, excluding rotating universes and propagations of the center of mass, respectively.\footnote{Arbitrary solutions of Newtonian mechanics can be projected onto shape space, but the total angular momentum (let's say) cannot be captured by relational initial data. It corresponds to a particular choice of absolute spatio-temporal reference frame (``gauge'') in the shape space theory; cf. Dürr et al. (2019). $\textbf{L} =0$ is then the only canonical choice, and the only one suggested by Machian principles.} Furthermore, the rejection of absolute time-scales leads to considering only universes with zero total energy ($E \equiv 0$), since this is the only value invariant under a rescaling of time-units. 

The problematic issue when it comes to Newtonian gravity is its lack of scale-invariance. Newtonian gravity has models that do not rotate ($L\equiv 0$) and models that do not propagate ($P \equiv 0$) but it does not have models that do not expand ($D:=\frac{1}{2}\dot{I} = \sum_{i=1}^N \textbf{q}_i\cdot\textbf{p}_i \equiv 0$; Barbour calls $D$ the \emph{dilatational momentum}). The characteristic size of an $N$-particle system is given by $\sigma=\sqrt{\frac{I}{m}}$, where $I=m \sum\limits_{i=1}^N \mathbf{q}_i^2 $ is the center-of-mass moment of inertia, and we have already seen that $I$ can never be constant for non-negative energy (equation \eqref{lagrange-jacobi}) but is roughly parabolic as a function of time. In other words: an $N$-particle universe interacting by Newton's law of gravity always changes in size.\footnote{In fact, the general Lagrange-Jacobi identity shows that $E=0$ and  $I\equiv const.$ is possible only if the potential is homogeneous of degree $-2$. This had motivated the alternative, scale-invariant theory of gravitation proposed in Barbour (2003). Here, we discuss the relational formulation of Newtonian gravity with the familiar $\frac{1}{r}$-potential.} Of course, we can (and will) insist that this is meaningless from a relational point of view -- since absolute distance is meaningless -- but the process has nonetheless dynamical (and thus empirical) consequences in Newtonian theory. Simply put: for constant energy, a gravitating system that expands also slows down.\footnote{Compare: The Newtonian laws are invariant under time-independent rotations, but particles in a \emph{rotating} universe experience a centrifugal force affecting their motions (see Maudlin (2012) for an excellent discussion of this issue going back to the famous Leibniz-Clarke debate).} Conversely, if we eliminate scale by hand, namely by a time-dependent coordinate transformation $q \to \frac{q}{\sigma(t)} $, the resulting dynamics can be formulated on shape space, but will no longer have the standard Newtonian form. Instead, the dynamics become non-autonomous (time-dependent) with scale acting essentially like friction (Barbour et al., 2014). 

How to capture this time-dependence without reference to external time? Barbour et al. make use of the fact that the dilatational momentum $D =\frac{1}{2}\dot{I}$ is monotonically increasing ($\ddot{I}>0$ by equation \eqref{lagrange-jacobi}) and can thus be used as an internal time-parameter, a kind of universal clock. In particular, we observe that $D=0$ precisely when $I$ reaches its global minimum. This \emph{central time} thus marks the mid-point between a period of contraction and a period of expansion, or better (though this remains to be justified): the Janus point between two periods of expansion with respect to opposite arrows of time. It provides, in particular, a natural reference point for parametrizing solutions of the shape space theory in terms of \emph{mid-point data} on the shape phase space $T^*S$.\footnote{Mathematically, this is the cotangent bundle of shape space $S$, just as Hamiltonian phase space is the cotangent bundle of Newtonian configuration space.}




There is one last redundancy from the relational point of view that Barbour et al. (2015) call \emph{dynamical similarity}. It comes from the invariance of the equations of motion under a simultaneous rescaling of internal time $D$ and shape momenta. More simply put: two solution trajectories are physically identical if they correspond to the same geometric curve in shape space, the same sequence of shapes, even if that curve is run through at different ``speeds''. Thus, factoring out the absolute magnitude of the shape momenta at central time, we reduce the relevant phase space (that parametrizes solutions) by one further dimension. The resulting space $PT^*S$ (mathematically, this is the projective cotangent bundle of shape space $S$) is \emph{compact}, which means, in particular, that it has a \emph{finite total volume} according to the uniform volume measure.
And this is where the relational formulation, i.e., the elimination of absolute degrees of freedom, really starts to pay off. Since the uniform measure on  $PT^*S$ -- that Barbour et al. take to be the natural typicality measure, following Laplace's principle of indifference, is normalizable -- it allows for a statistical analysis that avoids the ambiguities resulting from the infinite phase space measure in the Carroll model. It should be noted that the construction of the measure is not entirely canonical; it involves the choice of a metric on shape space (which, in turn, can be defined through a scale-invariant metric on absolute phase space). 
In general, the justification of the typicality measure remains a critical step that would require a more in-depth discussion.\footnote{Dürr, Goldstein, and Zangh\`i (2019) provide an insightful discussion of typicality measures on shape space, though focussing on the quantum case.} For instance, we are rather skeptical of a ``principle of indifference'' as a motivation for a uniform measure.\footnote{But see the contribution of Bricmont to this volume, which defends a modern version of the principle.} For the time being, we just take the pragmatic attitude that the typicality measure proposed by Barbour, Koslowski, and Mercati will be justified by its empirical and explanatory success. Deviating from their notation, we denote this measure on the reduced phase space by $\mu_\varepsilon$. 

\subsection{Shape complexity and the gravitational arrow }
To describe the macro-evolution of a gravitating system on shape space, Barbour and collaborators introduce a dimensionless (scale-invariant) macrovariable $C_S$ which they call \emph{shape complexity}:
\begin{equation}C_S=-V\cdot \sqrt{I}.\end{equation} 

\noindent Comparison with \eqref{phasespacevolume} (setting $E=0$ and remembering that $U=-V$) shows a remarkable relationship between this shape complexity and the gravitational entropy that we computed on absolute phase space ($S(E=0, I, U) \approx N \log(\sqrt{I}\cdot C_S)$). Recalling our previous discussion (or noting that $C_S\approx R/r$, where $R$ is the largest and $r$ the smallest distance between particles), we also see that low shape complexity corresponds to dense homogeneous states in absolute space, while high shape-complexity indicates ``structure'' -- dilute configurations of multiple clusters. 

On shape space, considering the simplest case of $N=3$ particles, the configuration of minimal shape complexity is the equilateral triangle, while the configuration of maximal shape complexity corresponds to ``binary coincidences'' in which the distance between two particles -- relative to their distance to the third -- is zero. This is to say that 3-particle configurations with high shape complexity will, in general, contain a Kepler pair (a gravitational bound state of two particles) with the third particle escaping to infinity. 

In Section \ref{sec:typicalevol}, we have discussed the typical evolution of $-V\cdot I$ and found it to be roughly parabolic or  $\mathsf{U}$-shaped. Analogously, one can conclude that the evolution of $C_S = -V \sqrt{I}$ (in Newtonian time) will typically exhibit a $\mathsf{V}$-shaped profile: it has a global minimum at the central time ($D=0$), from which it grows roughly linearly in both time directions (see Fig. \ref{fig:shapecomplexity}). In the terminology of Barbour, Koslowski, and Mercati, this defines two opposite \emph{gravitational arrows of time} with the Janus point as their common past. Note that these are not \emph{entropic} arrows, though our previous discussion strongly suggests that the evolution of the shape complexity on shape space will align with the evolution of the gravitational entropy \eqref{graventropy} on absolute space. 

A remarkable feature of the relational theory, however, is that it reveals the origin of the gravitational arrow to be \emph{dynamical} rather than \emph{statistical}: the negative of the shape complexity corresponds to the potential that generates the gravitational dynamics on shape space. There is thus a dynamical tendency towards higher values of $C_S$ (lower values of the shape potential). In contrast, Boltzmannian statistical mechanics suggests that entropy increase is typical because there are a great many more high-entropy than low-entropy configurations that a system could evolve into. It does not suggest that physical forces are somehow driving the system towards higher entropy states.

\begin{figure}[ht]\label{fig:shapecomplexity}
\begin{center}
\includegraphics[width=0.9\textwidth]{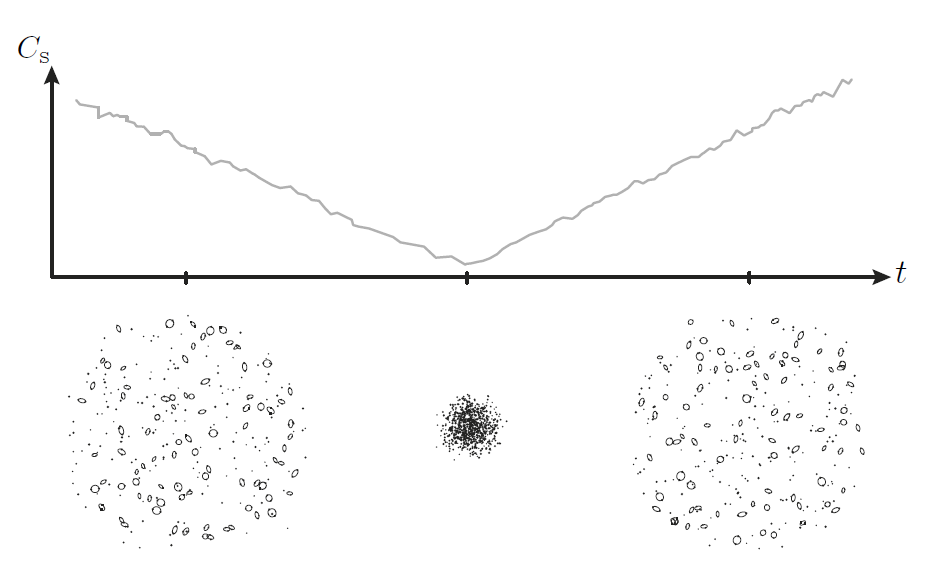}
\caption{\footnotesize Top: evolution of the shape complexity $C_S$ found by numerical simulation for $N = 1000$ and Gaussian initial data. Bottom: schematic conception (not found by numerical simulation) of three corresponding configurations on Newtonian space-time. Source: Barbour et al. (2015).}\label{fig:cs}
\end{center}
\end{figure}

Turning to the statistical analysis of the shape space theory, we are interested in determining typical values of $C_S$ at the Janus point. To this end, we consider the  measure assigned to mid-point data (Janus point configurations) with low shape complexity
\begin{equation}C_S \in [C_{min}, \,  \alpha\cdot C_{min}]:=I_1 , \end{equation} 
respectively high shape complexity
\begin{equation}C_S \in (\alpha\cdot C_{min}, \infty) :=I_\infty.\end{equation}
Here, $1<\alpha \ll \infty$ is some positive constant and $C_{min}$ is the smallest possible value of $C_S$. The key result of the relational statistical analysis (not yet rigorously proven but strongly substantiated by the 3-particle case and numerical experiments for large $N$) is now that already for small values of $\alpha$ ($\alpha < 2$ for large $N$)
\begin{equation} \frac{\mu_\varepsilon({I_\infty})}{\mu_\varepsilon(PT^*S)} \approx 0,\end{equation} 
and consequently
\begin{equation} \frac{\mu_\varepsilon( I_1)}{\mu_\varepsilon(PT^*S)} \approx 1. \end{equation} 
This means: it is typical that a universe at the Janus point (the beginning of our macro-history) is in a very homogeneous state!

Regardless of the philosophical merits of relationalism, the shape space theory of Barbour, Koslowski, and Mercati thus comes with two great virtues: First, it provides a sensible \emph{normalizable} measure on the set of possible micro-evolutions that still establishes an arrow of time as typical. Even more spectacularly, typical evolutions with respect to this measure go through very homogeneous configurations at their Janus point ($\sim$ the ``big bang''). In other words, initial states that have very low entropy from the absolutist point of view come out as typical in the shape space description -- provided that one accepts the proposed measure on $PT^*S$ as a natural typicality measure. This would resolve the two potential problems that we have identified for the Carroll model: the mysteriously low entropy of our universe and the justification for locating our present state reasonably far away from the entropy minimum. 

\subsection{Entaxy and Entropy}

On the other hand, Barbour et al. introduce another concept called (instantaneous) \emph{entaxy} that we find much less compelling. The instantaneous entaxy (the authors also use the term \emph{solution entaxy} for the measure $\mu_\varepsilon$ on $PT^*S$) is supposed to be the measure of a set of shape configurations corresponding to a given value of shape complexity. It thus seems \emph{prima facie} analogous to the Boltzmann entropy defined in terms of the macrovariable $C_S$, with the notable exception that it \emph{decreases} in time as the shape complexity increases. Recall, however, that the measure $\mu_\varepsilon$ was only defined on mid-point data, by cutting through the set of solution trajectories at their Janus points, so to speak. Barbour et al. now extend it to arbitrary (internal) times by stipulating that the entaxy associated with a particular value of shape complexity at \emph{any} point in history is the measure of \emph{mid-point} configurations with that value of shape complexity. 

This definition seems somewhat \emph{ad hoc} and corresponds to comparing macrostates at different times in terms of a measure that is \emph{not stationary} under the dynamics: A set of mid-point data will have a bigger size than the set of time-evolved configurations (phase space volume gets lost, so to speak). Indeed, on the 3-particle shape space, it is not hard to see that the points of maximal shape complexity are dynamical attractors; hence, a stationary continuous measure on shape phase space does not exist. In general, it is not even clear if a stationary measure is a meaningful concept in relational mechanics since there is no absolute (metaphysical) sense in which configurations on different solution trajectories with the same internal time are \emph{simultaneous}. They merely happen to agree on whatever part or feature of the configuration plays the role of an internal ``clock''. For all these reasons, the entaxy should not be regarded as a shape analogon of the Boltzmann entropy (which is always defined in terms of a stationary measure). In particular, the fact that the gravitational arrows point in the direction of decreasing rather than increasing entaxy is not in contradiction with Boltzmannian arguments. 

Finally, one may wonder whether we could compute on absolute phase space the Boltzmann entropy associated to the shape complexity or other scale-invariant macrovariables. Note that for $E=0$, our gravitational entropy \eqref{graventropy} is a function of $-V I$. Couldn't we have just computed an entropy for the macrovariable $C_S= -V \sqrt{I}$, instead? Interestingly, the answer is negative, and the reason is the following simple result, showing that macroregions defined by scale-invariant macrovariables would have measure zero or infinity.

\begin{Proposition}
	Let $\mu$ a measure on $\IR^n$ (equipped with its Borel sigma-algebra), which is homogeneous of degree $d$, i.e., $\mu(\lambda A)= \lambda^d \mu (A)$ for any measurable $A \subset \IR^n$ and all $\lambda >0$. Let $F: \IR^n \to \IR^m$ be a measurable function homogeneous of degree $k$, i.e., $F(\lambda x) = \lambda^k F(x), \; \forall x \in \IR^n$. Then we have for any measureable value-set $J \subset \IR^m$:
	\begin{equation}\label{eq:hommeasure}
	\mu\left(\lbrace x \mid F(x) \in  \lambda^k J \rbrace\right) = \lambda^d \mu\Bigl(\lbrace x \mid F(x) \in J \rbrace\Bigr) .
	\end{equation}
	\begin{proof}
		\begin{align*}
		\mu\left(\lbrace x \mid F(x) \in \lambda^k J \rbrace\right) =\mu\left(\lbrace \lambda x \mid F(x) \in J \rbrace\right) = \lambda^{d} \mu\left(\lbrace x \mid F(x) \in  J \rbrace\right).
		\end{align*}
	\end{proof}
\end{Proposition}

\noindent From this, we can immediately conclude:
\begin{Corollary}
	If the measure $\mu$ is homogeneous of degree $d \neq 0$ and $F$ is  homogeneous of degree $0$ (i.e., scale-invariant), then
	\begin{equation} 
	\mu\left(F^{-1}(J)\right) \in \lbrace 0, +\infty\rbrace.
	\end{equation}
	\begin{proof}
		Applying \eqref{eq:hommeasure} with $k =0$ and $d\neq 0$ yields $\mu\left(F^{-1}(J)\right) = \lambda^d \mu\left(F^{-1}(J)\right)$ for any $\lambda >0$.
	\end{proof}
\end{Corollary}

\noindent Hence, using a homogeneous phase space measure (such as the Liouville measure or the microcanonical measure for a homogeneous potential and $E=0$) macroregions defined in terms of scale-invariant macrovariables must have measure zero or infinity, so that the corresponding Boltzmann entropy would be ill-defined. This suggests that the concept of entropy is intimately linked to absolute scales and thus not manifestly relational. Note, in particular, that \emph{expansion} and \emph{heating} -- processes that are paradigmatic for entropy increase (especially, but not exclusively in our analysis of gravitating systems) -- require absolutes scales of distance and velocity, respectively. 

This emphasizes once again that the gravitational arrow of Barbour et al. is not an entropic arrow, although it matches -- maybe accidentally, maybe for reasons we don't yet understand -- the entropic arrow that we identified on absolute phase space. The result also leaves the relationalist with the following interesting dilemma: Either the notion of entropy is meaningful only for subsystems -- for which the environment provides extrinsic scales -- or we have to explain why the entropy of the universe is a useful and important concept \emph{despite} the fact that it is related to degrees of freedom that are strictly speaking unphysical, corresponding to mere gauge in the shape space theory.

\section{Conclusion}
The works of Carroll and Chen as well as Barbour, Koslowski, and Mercati show that it is possible to establish an arrow of time as typical, without the need to postulate special boundary conditions or any other form of Past Hypothesis. By proposing the definition of a Boltzmann entropy for a classical gravitating universe, we argued that Newtonian gravity provides a relevant realization of Carroll's entropy model that can be compared to the shape space formulation of Barbour et al. We found, in particular, that the gravitational arrows identified by Barbour and collaborators in terms of shape complexity will match the entropic arrows in the theory on absolute space. The extension to other microscopic theories (and/or macroscopic state functions) will require further research. The relationalist and the absolutist approaches both provide the resources to avoid the reversibility paradox and ground sensible inferences about the past and future of the universe. However, while certain ambiguities remain in the Carroll model, resulting from the non-normalizability of the phase space measure, those issues are resolved by the shape space theory of Barbour et al. -- provided one accepts their choice of typicality measure. In any case, for a Newtonian gravitating universe, their analysis suggests that homogeneous configurations at the ``big bang'' (Janus point) are \emph{typical}, explaining why the universe started in what looks like a very low-entropy state from an absolutist perspective. However, if the shape space theory is actually fundamental, the ``entropy of the universe'' turns out to be a somewhat spurious concept whose status remains to be discussed in more detail.\\ 

\noindent \textbf{Acknowledgements:} We thank Julian Barbour, Detlef Dürr, Valia Allori, Sheldon Goldstein, and two anonymous referees for valuable comments. 

\newpage
\appendix
\section*{Appendix: Computation of the gravitational entropy}
We compute the phase space volume of the macro-region $\Gamma(E,I\pm \epsilon I,U \pm \epsilon U)$, i.e.\footnote{Technically, we have to start with $\int\int ... \mathds{1}\Bigl\lbrace  (1-\epsilon)I \leq \sum\limits_{i=1}^N m (\textbf{q}_i-\frac{1}{N}\sum\limits_{j=1}^N \textbf{q}_j)^2 \leq (1+\epsilon)I\Bigr\rbrace ...$ and the additional constraint $\sum\limits_{j=1}^N \textbf{q}_j \equiv 0$ fixing the center of mass to the origin of our coordinate system.} : 
\begin{equation} \begin{split}\label{entropyintegral1}
\frac{1}{N!h^{3N}} \int d^{3N}p &\int d^{3N}q\ \delta \big(H(\textbf{q},\textbf{p})-E\big)\\ \mathds{1} &\Bigl\lbrace (1-\epsilon)U  \leq  \sum_{\substack{i<j\\ i,j=1}} \frac{Gm^2}{|\textbf{q}_i-\textbf{q}_j|} \leq (1+\epsilon)U \Bigr\rbrace \mathds{1}\Bigl\lbrace  (1-\epsilon)I \leq \sum_{i=1}^N m\textbf{q}_i^2 \leq (1+\epsilon)I\Bigr\rbrace,\end{split}\end{equation}
with the Hamiltonian
\begin{align*} H(q,p) = \sum\limits_{i=1}^N \frac{\mathbf{p}^2_i}{2m} \, - \sum\limits_{1\leq i < j \leq N} \frac{Gm^2}{\lvert \mathbf{q}_i - \mathbf{q}_j\rvert}.\end{align*}
We shall prove the following
\begin{Proposition*}
	If the scale-invariant quantity $\sqrt{I} U$ is large enough that
	\begin{equation}\label{IUbound}
	\sqrt{I} U \geq 4Gm^{5/2} \log(N)N^{5/2} \,,
	\end{equation}
	we obtain the bounds:
	\begin{align*}
	&\lvert \Gamma(E,I\pm \epsilon I,U \pm \epsilon U) \rvert \geq  C  \,  e^{-\frac{9}{2}N} \Bigl(\frac{G m^2 \epsilon}{U}\Bigr)^3 \left(E+(1-\epsilon)U\right)^{\frac{3N-2}{2}}  \Bigl(\frac{I}{m}\Bigr)^{\frac{3N-3}{2}} \, ,\\
	&\lvert \Gamma(E,I\pm \epsilon I,U \pm \epsilon U) \rvert   \leq  {C} \left(E+(1+\epsilon)U\right)^{\frac{3N-2}{2}} \Bigl((1+\epsilon) \frac{I}{m} \Bigr)^{\frac{3N}{2}} \, , 
	\end{align*}
	for any $1> \epsilon > \frac{2}{N}$ and $N\geq 4$, where $C=\frac{(2m)^{\frac{3N-2}{2}}}{2  (N!)h^{3N}} \,\left(\Omega^{3N-1}\right)^2$, with $\Omega^{3N-1}$ the surface area of the $(3N-1)$-dimensional unit sphere.\\
	
	\noindent For non-negative $E$, this can be simplified further using $(E+(1+\epsilon)U)^{n} \leq (1+\epsilon)^{n} (E+U)^{n} \leq e^{\epsilon n} (E+U)^{n}$, respectively  $(E+(1-\epsilon)U)^{n} \geq (1-\epsilon)^{n}(E+U)^{n}  \geq  e^{-2\epsilon n} (E+U)^{n}$, for $\epsilon <\frac{1}{2}$. 
\end{Proposition*}

\begin{proof}
	We first perform the integral over the momentum variables and are left with
	\begin{align*} \frac{(2m)^{\frac{3N-2}{2}}}{2 N!h^{3N}} \,&\Omega^{3N-1} \int \mathrm{d}^{3N}q \, \bigg(E+\sum_{i<j}\frac{Gm^2}{|\textbf{q}_i-\textbf{q}_j|}\bigg)^{\frac{3N-2}{2}}\\ &\mathds{1} \Bigl\lbrace (1-\epsilon)U  \leq  \sum_{\substack{i<j\\ i,j=1}} \frac{Gm^2}{|\textbf{q}_i-\textbf{q}_j|} \leq (1+\epsilon)U \Bigr\rbrace \mathds{1}\Bigl\lbrace  (1-\epsilon)I \leq \sum_{i=1}^N m\textbf{q}_i^2 \leq (1+\epsilon)I\Bigr\rbrace.\end{align*}
	From this, it is straightforward to obtain the upper bound: 
	\begin{equation*}\begin{split}
	\eqref{entropyintegral1}&\leq  \frac{(2m)^{\frac{3N-2}{2}}}{2 N!h^{3N}} \,\Omega^{3N-1}  \frac{2\Omega^{3N-1}}{3N} (E+(1+\epsilon)U)^{\frac{3N-2}{2}}  \Bigl[ \Bigl((1+\epsilon)\frac{I}{m}\Bigr)^{\frac{3N}{2}}  -((1-\epsilon)\frac{I}{m})^{\frac{3N}{2}}  \Bigr] \\
	&\leq   \frac{C}{3N}  \left (1+\epsilon\right)^{\frac{3N}{2}} \bigl(E+(1+\epsilon)U\bigr)^{\frac{3N-2}{2}} \Bigl(\frac{I}{m} \Bigr)^{\frac{3N}{2}}. 
	\end{split}\end{equation*}

	\noindent For the lower bound, we consider the set $\mathcal{B}:= B_1 \times \ldots \times B_{N} \subset \IR^{3N}$ defined by
	\begin{equation*}B_j:= \bigl\lbrace \lvert \mathbf{q}_j \rvert \in [(2j-2) \xi , (2j-1) \xi] \bigr\rbrace , \end{equation*}
	with $\xi = \xi(I, N)$ to be determined soon. That is, we consider a series of concentric spheres around the origin, their radii being an increasing multiple of $\xi$, and configurations for which the volume between two spheres is alternately empty or occupied by a single particle. In  $\mathcal{B}$, we have $mq^2 \in [I_{+}-\Delta I, I_{+}]$ with 
	\begin{equation}
	I_{+}=\sum\limits_{i=2}^{N} m \mathbf{q}^2_i \leq m \sum\limits_{i=2}^{N}(2i-1)^2 \xi^2 = \frac{m}{3}N(4N^2 -1)\xi^2
	\end{equation}
	and
	\begin{equation*}
	\Delta I = m\sum\limits_{i=1}^{N}[(2i-1)^2 - (2i-2)^2]\xi^2 =m \sum\limits_{i=1}^{N}[4i-3]\xi^2 \leq 2m N^2  \xi^2 \leq \frac{2}{N} I_{+}.
	\end{equation*}
	We thus set 
	\begin{equation}
	\xi := \sqrt{\frac{3 I}{m N (4N^2 -1)}},
	\end{equation}
	so that $I_+=I$, and note that $\frac{2}{N}I < \epsilon I$ for $\epsilon > \frac{2}{N}$ , so configurations in $\mathcal{B}$ are within the right range of values for the moment-of-inertia macrovariable. 
	
	Now we have to make sure to consider configurations whose potential energy is also in the right range $\rvert V(q)\rvert \in [U \pm \epsilon U]$. 
	
	To this end, we note that for $q \in \mathcal{B}$, the distance between two particles is bounded from below by $\lvert \mathbf{q}_i - \mathbf{q}_j\rvert \geq (2(j-i)-1)\xi$, and for each $1\leq k \leq N-1$, there exists less than $N$ particle pairs with $j-i=k$. Hence, the potential energy is bounded by 
	\begin{equation}
	\sum\limits_{i<j}\frac{Gm^2}{\lvert \mathbf{q}_i - \mathbf{q}_j\rvert} \leq N \sum\limits_{k=1}^{N-1} \frac{Gm^2}{(2k-1)\xi}\leq  2 N \log(N) \frac{Gm^2}{\xi} < U,
	\end{equation}
	where we used the assumption \eqref{IUbound} and the estimate
	\begin{equation*}
	\sum\limits_{k=1}^{N-1}\frac{1}{(2k-1)} \leq 1 + \sum\limits_{k=2}^{N-1}\frac{1}{k} \leq  1 + \int\limits_1^{N} \frac{1}{x} \dd x = 1 +\log(N) \leq 2\log(N), \; \text{for } N \geq 4.
	\end{equation*}
	\noindent In particular, we know that for $q \in \mathcal{B}$ and, e.g., $\mathbf{q}_1 =\mathbf{0}$,  we have $\lvert V(q)\rvert < U$. But also that $\lim_{ \mathbf{q}_1 \to \mathbf{q}_2} \lvert V(q)\rvert = +\infty$. Hence, by the mean value theorem, there exists for any $(\mathbf{q}_2, \ldots, \mathbf{q}_N) \in B_2\times \ldots \times B_N$ a $\lambda \in (0,1)$ such that $-V(\lambda \mathbf{q}_2, \mathbf{q}_2, \ldots, \mathbf{q}_N) = U$.
	
	Now it is not difficult to check that if we place $\mathbf{q}_1$ at a distance of not more than $r:=\epsilon \frac{Gm^2}{2U}$ from $\lambda \mathbf{q}_2$, the potential energy $\lvert V(q)\rvert =  \sum\limits_{i=2}^N \frac{Gm^2}{\lvert \vec{q}_1 - \vec{q}_i\rvert} + \sum\limits_{1\neq i < j \leq N} \frac{Gm^2}{\lvert \vec{q}_i - \vec{q}_j\rvert}$ will change by less than $\left\lvert \lvert V(q)\rvert - U\right\rvert \leq \epsilon U$. Moreover, while $\mathbf{q}_1$ may no longer be in $B_1$, the moment of inertia $m\sum_{i=1}^{N} \mathbf{q}_i^2$ increases by less than $m(3\xi)^2 < \epsilon I$ and thus remains within the interval $[I \pm \epsilon I]$.

	
	We denote by $K_{r}[q]$ the ball of radius $r$ around $\lambda\mathbf{q}_2$, with $\lambda = \lambda(\mathbf{q}_2, \ldots, \mathbf{q}_N)$ as introduced above. Its volume is $\frac{4\pi}{3} \Bigl(\frac{Gm^2\epsilon}{2U}\Bigr)^3$. The volume of the set $B_j, j\in \lbrace 2,\ldots, N\rbrace$ is
	\begin{equation*}
	\lvert B_j \rvert = \frac{4 \pi}{3}\left[ (2j-1)^3 - (2j-2)^3   \right] \xi^3 = \frac{4 \pi}{3}\left[ j(12j-18)+7 \right] \geq \frac{16 \pi}{3}j^2 \, .
	\end{equation*}
	Hence:
	\begin{equation*}
	\lvert B_2 \times \ldots \times B_N \rvert = \prod_{j=2}^{N} \lvert B_j \rvert = \Bigl(\frac{16 \pi}{3}\Bigr)^{N-1}  \xi^{3(N-1)} \prod_{j=2}^{N} j^2 = \Bigl(\frac{16 \pi}{3}\Bigr)^{N-1}  \xi^{3(N-1)} (N!)^2 \, .
	\end{equation*}
	To simplify, we note that $\xi^{3(N-1)} \geq \Bigl(\frac{3}{4}\Bigr)^{\frac{3(N-1)}{2}}  \Bigl(\frac{I}{m}\Bigr)^{\frac{3(N-1)}{2}} N^{-{\frac{9(N-1)}{2}}}$. And  comparing with the area of the unit sphere,  $\Omega^{3N-1}=\frac{2 \pi^{3N/2}}{\Gamma\left(\frac{3N}{2}\right)}$, we use $\Bigl(\frac{16 \pi}{3}\Bigr)^{N-1} \Bigl(\frac{3}{4}\Bigr)^{\frac{3(N-1)}{2}}  > \Omega^{3N-1}\Gamma\left(\frac{3N}{2}\right)$. Thus: 
	\begin{align*} \eqref{entropyintegral1} \geq &\,
	\frac{(2m)^{\frac{3N-2}{2}}}{2 N!h^{3N}} \,\Omega^{3N-1} \int\limits_{K_{r(q)}\times B_2 \times \ldots \times B_N} \mathrm{d}^{3N}q \, \bigg(E+\sum_{i<j}\frac{Gm^2}{|\textbf{q}_i-\textbf{q}_j|}\bigg)^{\frac{3N-2}{2}} \mathds{1} \Bigl\lbrace \ldots\Bigr\rbrace \mathds{1} \Bigl\lbrace \ldots\Bigr\rbrace\\
	\geq & \,  C\, N^{-{\frac{9(N-1)}{2}}} (N!)^2 \, \Gamma\left(\frac{3N}{2}\right)  \Bigl(\frac{G m^2 \epsilon}{2U}\Bigr)^3 \  \left(E+(1-\epsilon)U\right)^{\frac{3N-2}{2}}  \Bigl(\frac{I}{m}\Bigr)^{\frac{3N-3}{2}} \, . \end{align*}
	
	\noindent Summing over all possible permutations of the particles over the rings in the definition of $\mathcal{B}$, we get an additional factor of $N!$. With the Sterling approximation $\Gamma(n+1)=n!>\sqrt{2\pi n} \Bigl(\frac{n}{e}\Bigr)^n$, we finally obtain a lower bound of the form
	\begin{equation*} 
	\eqref{entropyintegral1} \geq C\, e^{-\frac{9}{2}N} \Bigl(\frac{ G m^2 \epsilon}{U}\Bigr)^3 (E+(1-\epsilon)U)^{\frac{3N-2}{2}}  \Bigl(\frac{I}{m}\Bigr)^{\frac{3N-3}{2}}.
	\end{equation*}
\end{proof}

\end{document}